\documentclass[a4paper,10pt, reqno]{amsart}
\usepackage[utf8]{inputenc}

\usepackage{amsmath,amssymb,amsfonts}
\usepackage{mathptmx}
\usepackage{varioref}
\usepackage{a4wide}
\usepackage{hyphenat}
\usepackage{multirow}
\usepackage{mathrsfs}
\usepackage{txfonts}
\usepackage[numbers,sort&compress]{natbib}

\usepackage{enumerate}
\usepackage{graphicx}

\usepackage[pdfpagelabels,hypertexnames=true,
            plainpages=false,
            naturalnames=false]{hyperref}
\usepackage{color}
\definecolor{darkblue}{rgb}{0,0.1,0.5}
\hypersetup{colorlinks,
            linkcolor=darkblue,
            anchorcolor=darkblue,
            citecolor=darkblue}

\numberwithin{equation}{section}

\newcommand{\bx}{\boldsymbol{x}}
\newcommand{\by}{\boldsymbol{y}}
\newcommand{\bxi}{\boldsymbol{\xi}}

\newcommand{\bi}{\boldsymbol{i}}
\newcommand{\fitMC}{\boldsymbol{f}}

\newcommand{\symgr}{\mathcal{S}}

\newcommand{\dd}{\mathrm{d}}
\newcommand{\N}{\mathbb{N}}
\newcommand{\Q}{\mathbb{Q}}
\newcommand{\Pb}{\mathbb{P}}

\newcommand{\argmin}{\operatorname{argmin}}
\newcommand{\1}{\boldsymbol{1}}
\newcommand{\ii}{\mathrm{i}}

\renewcommand{\leq}{\leqslant}
\renewcommand{\geq}{\geqslant}

\usepackage[capitalise]{cleveref}

\crefname{table}{Table}{Tables}
\creflabelformat{enumi}{#2#1#3)}

\newtheorem{theorem}{Theorem}
\newtheorem{lemma}{Lemma}

\newtheorem{proposition}{Proposition}

\theoremstyle{remark}
\newtheorem{remark}{Remark}
\newtheorem*{conjecture}{Conjecture}

\title{Asymptotic fitness distribution in the Bak--Sneppen model of biological evolution with four species}
\author{Eckhard Schlemm}

\begin{document}

\address{Wolfson College, University of Cambridge}
\email{es555@cam.ac.uk}


\begin{abstract}
We suggest a new method to compute the asymptotic fitness distribution in the Bak--Sneppen model of biological evolution. As applications we derive the full asymptotic distribution in the four\hyp{}species model, and give an explicit linear recurrence relation for a set of coefficients determining the asymptotic distribution in the five\hyp{}species model.
\end{abstract}

\maketitle

\section{Introduction and main result}
The Bak--Sneppen (BS) model of biological evolution has been introduced in \cite{Bak1993punctuated} to study self-organised critical behaviour in populations with natural selection and spatial interactions. In the model, a population is spread out over a circle so that each species has exactly two neighbours, and each site is assigned a numerical parameter between zero and one that describes the {\it fitness} of the species at that particular site. The system evolves according to the following discrete\hyp{}time dynamics: at each time step the species with the smallest fitness parameters becomes extinct and is replaced by a new species whose fitness parameter becomes independent uniformly distributed on the interval $[0,1]$. To take into account interactions between species, e.\,g.\ competition for resources or predator\hyp{}prey relations, the fitness parameters of the two neighbours of the least fit species are also reset to random values.
This represents the idea that the fitness is not just a fixed property of a species but also depends on its environment. Even though this is a very rudimentary model for evolution it does possess some of the characteristic features thought to be associated with evolving biological systems, including long\hyp{}range dependence and evolutionary activity on all time scales, see e.\,g.\ \citep{jensen1998self} for a discussion in the context of self\hyp{}organised criticality.

Despite its apparently easy definition, the BS model has withstood most attempts at mathematical analysis in the past. Partial results have been obtained, however, in the context of rank\hyp{}driven processes and mean\hyp{}field approximations \citep{flyvbjerg1993mean,grinfeld2011bak,grinfeld2011rank}. Early on it was conjectured based on simulations that the fitness distribution at a fixed site converges, in the limit of large populations, to a uniform distribution on the interval $[f_c,1]$, where $f_c\approx 0.667$.

There are only few rigorous results about the original BS model; in \citep{MeesterZnamenski2003limit} it is proved that the expected fitness at a fixed site in the stationary regime is bounded away from unity uniformly in the size of the population; \citep{MeesterZnamenski2004critical} gives a conditional characterization of the limiting distribution in terms of a set of critical thresholds \citep[see also][]{gillet2006bounds,gillet2006maximal}.

In this paper we suggest a novel approach to determining the stationary fitness distribution in the BS model with a finite number of species. The method consists of translating the evolutionary dynamics into a set of linear recurrence equations for the coefficients of the densities of the finite\hyp{}time distributions and analysing the asymptotic properties of these recursive equations. We apply the method to derive the asymptotic fitness distribution in the BS model with four species, which is given by (\cref{theorem-asympdist4})
\begin{equation*}
\mu(\dd^4\bx)=\frac{3}{2}\left[\sum_{\nu=1}^4{\frac{x_\nu(3-x_\nu(3-x_\nu))}{(3-x_\nu(3-x_\nu(3-x_\nu)))^2}}\right]\1_{[0,1]^4}(\bx)\dd^4\bx.
\end{equation*}

It has been suggested that a similar, and potentially easier, analysis can be applied to the anisotropic Bak--Sneppen (aBS) model; in this variation of the model the evolutionary dynamics are modified so that the fitness parameter of only one neighbour, say the right one, is updated at each time step together with the fitness parameter of the least fit species. While many qualitative properties are believed to be shared by the original and the anisotropic BS, analytical considerations and numerical calculations indicate that they fall in different universality classes, and have thus different critical exponents \citep{head1998anisotropic, maslov1998critical}. Our method can be applied to the aBS. In the first non\hyp{}trivial case of three species the asymptotic fitness distribution follows from computations which are very similar to those leading to \cref{theorem-asympdist4}, and is given by
\begin{equation*}
\mu(\dd^3\bx)=\frac{2}{3}\left[\sum_{\nu=1}^3{\frac{x_\nu(2-x_\nu)}{(2-x_\nu(2-x_\nu))^2}}\right]\1_{[0,1]^3}(\bx)\dd^3\bx.
\end{equation*}
The complexity of the calculations for larger populations seems to be increasing at a similar rate as in BS and studying aBS does therefore not, for the purpose of understanding the method we propose in the current paper, offer major advantages over the original model. For this reason we concentrate in the following on the isotropic case.

\paragraph{Outline of the paper}
The paper is organised as follows. In \cref{section-qualana} we present one of several possible, precise definitions of the BS model in terms of a Markov chain, and use this representation to derive qualitative properties for arbitrary population sizes. We then turn to the special case of the four\hyp{}species model in \cref{section-n4}, where we derive the full asymptotic fitness distribution. Finally, \cref{section-genrec} is devoted to making further progress towards a similarly explicit solution for larger populations. In particular, we obtain explicit linear recursions for the coefficients that describe the finite\hyp{}time distributions in the five\hyp{}species model.
\paragraph{Notation}
Throughout the paper we denote by $\N$ the set of positive integers and write $n\in\N$ for the number of species in the BS model. Arithmetic operations in indices pertaining to the number of species are always understood to be performed modulo $n$, so as to give a result between $1$ and $n$. The symmetric group on $n$ letters is denoted by $\symgr_n$. For a vector $\bx\in[0,1]^n$ and an integer $\nu\in\{1,\ldots,n\}$ we denote by $(\cdot)_{[\nu]}$ and $(\cdot)_{]\nu[}$ the inner and outer $\nu$\hyp{}centred projections, that is
\begin{align*}
(\cdot)_{[\nu]}:[0,1]^n\to[0,1]^3,\quad& \bx_{[\nu]} = (x_{\nu-1},x_\nu,x_{\nu+1});\\
(\cdot)_{]\nu[}:[0,1]^n\to[0,1]^{n-3},\quad& \bx_{]\nu[} =  (x_{\nu+2},\ldots,x_n,x_1,\ldots,x_{\nu-2}).
\end{align*}
The $i$th components of the vectors $\bx_{[\nu]}$ and $\bx_{]\nu[}$ are denoted by $\bx_{[\nu],i}$ and $\bx_{]\nu[,i}$, respectively. Given another vector $\bxi\in[0,1]^3$ we extend the notation of outer projections by setting
\begin{equation*}
(\cdot)_{]\nu[_{\bxi}}:[0,1]^n\to[0,1]^n,\quad \bx_{]\nu[_{\bxi}} = (\xi_1,\xi_2,\xi_3,x_{\nu+2},\ldots,x_n,x_1,\ldots,x_{\nu-2}).
\end{equation*}
We use the multi\hyp{}index notation $\bx^{\bi}$, $\bx\in[0,1]^n$, $\bi\in\N_0^n$, as a shorthand for the monomial $x_1^{i_1}\cdots x_n^{i_n}$. We further let $\Pb$ denote probability and write $\mathscr{B}(X)$ for the Borel sets of a topological space $X$. Finally, we denote by $\1_{\{\mathcal{E}\}}$ the indicator of an expression $\mathcal{E}$, which is defined to be one if $\mathcal{E}$ is true, and zero else; for a set $S$, we also write $\1_S(x)$ instead of $\1_{\{x\in S\}}$.

\section{Definition of the model and qualitative analysis for general population sizes}
\label{section-qualana}
We interpret the Bak--Sneppen model with $n$ species as a Markov chain $\fitMC_n=(\fitMC_{n,k})_{k\geq 0}$ on the hypercube $[0,1]^n$, see, e.\,g.\ \citep{meyn2009markov} for a general introduction to this topic. To avoid the trivial case in which all fitness parameters are updated at each time step, we assume that $n\geq 4$. For each time $k\in\N$, the values of $\fitMC_{n,k}=\left(f_{n,k,\nu}\right)_{\nu=1,\ldots,n}$ denote the fitnesses of the different species in the population. In order to formalise the heuristic description of the evolutionary processes described in the introduction we first observe that the fitness landscape $\fitMC_{n,k+1}$ at time $k+1$ can be determined easily if one knows the fitness values at time $k$ and, in particular, which species was the least fit at that time. In fact, if at time $k$ the $\nu$th species is the least fit, then the fitness values of the species at the $(\nu-1)$th, $\nu$th, and $(\nu+1)$th site are independent uniformly distributed at time $k+1$, while all other fitness values remain unchanged. For this argument to be valid, it is necessary that $\fitMC_{n,k}$ has a unique minimal value. In fact the BS model, as we have described it, is not well\hyp{}defined without a rule for breaking ties. One natural possibility is to randomly select one of the sites with minimal fitness value. We will, however, not have to concern ourselves with this complication because we will assume that the initial fitness distribution at time $k=0$ is absolutely continuous, in which case ties occur with probability zero. Further formalising the argument, we can partition the event $\{\fitMC_{n,k+1}\in A\}$, $A=\prod_\nu[x_\nu^-,x_\nu^+]\in\mathscr{B}([0,1]^n)$, as $\bigcup_{\nu=1}^n{\{\fitMC_{n,k+1}\in A\}\cap \{\min\fitMC_{n,k}=f_{n,k,\nu}\}}$ and obtain,
\begin{align}
\label{eq-recfnkAabstract}
\Pb\left(\fitMC_{n,k+1}\in A\right) =& \sum_{\nu=1}^n{\Pb\left(\fitMC_{n,k+1}\in A \text{ and } \min\fitMC_{n,k}=f_{n,k,\nu}\right)}\notag\\
  =& \sum_{\nu=1}^n{\Pb\left(\left(\fitMC_{n,k}\right)_{]\nu[}\in A_{]\nu[} \text{ and } \min\fitMC_{n,k}=f_{n,k,\nu}\right)\Pb\left((U_1,U_2,U_3)\in A_{[\nu]}\right)},
\end{align}
where $(\cdot)_{]\nu[}$ and $(\cdot)_{[\nu]}$ denote the inner and outer projections, respectively. The random variables $U_1,U_2,U_3$ are independent and uniformly distributed on $[0,1]$, and thus the last probability equals $\operatorname{Leb}^3\left(A_{[\nu]}\right)$. Put differently, the transition kernels $\Pb_{n,\bx}^1(\cdot)$ of the Markov chain $\fitMC_n$, which are characterized by the equations
\begin{equation*}
\Pb\left(\fitMC_{n,k+1}\in A\left|\fitMC_{n,k}=\bx\right.\right) = \int_A{\Pb_{n,\bx}^1(\dd^n\bxi)},\quad \bx\in[0,1]^n,\quad A\in\mathscr{B}([0,1]^n),
\end{equation*}
are given by
\begin{equation}
\label{eq-transkern}
\Pb_{n,\bx}^1(\dd^n\bxi) = \prod_{\mu\notin\{\nu-1,\nu,\nu+1\}}\delta_{x_\mu}(\dd\xi_\mu)\dd^3(\xi_{\nu-1},\xi_\nu,\xi_{\nu+1}),\quad \nu=\argmin\bx;
\end{equation}
here, $(\bxi,\bx)\in[0,1]^n\times[0,1]^n$, and $\delta_x(\cdot)$ denotes the Dirac measure located at $x$. This quite explicit expression is enough to prove that the Markov chains $\fitMC_n$ are ergodic in a strong sense, which is crucial for our approach to determining their stationary distributions. For the definition of uniform convergence we refer the reader to \cite[Definition (16.6)]{meyn2009markov}.
\begin{proposition}
\label{prop-ergodicity}
For every positive integer $n$, the Markov chain $\fitMC_n$ is uniformly ergodic.
\end{proposition}
\begin{proof}
It is straightforward, albeit tedious, to show that the multi\hyp{}step transition kernels $\Pb_{n,\bx}^m(\cdot)$, which are defined recursively by
\begin{equation*}
\Pb_{n,x}^{m}(A) = \int_{\bxi\in A}\int_{\by\in[0,1]^n}\Pb_{n,\bx}^1\left(\dd^n\by\right)\Pb_{n,\by}^{m-1}\left(\dd^n\bxi\right),\quad m\geq 2,
\end{equation*}
dominate, for each  $m\geq n/3-1$, a constant times Lebesgue measure restricted to $[0,1]^n$. The claim is then a consequence of \cite{Aldous1997mixing}[Theorem B, ii)]. Since a similar result has been obtained in \citep{gillet2007thesis}, we omit the technical details.
\end{proof}
\Cref{prop-ergodicity} implies in particular that there exists a unique invariant distribution $\mu_n$ satisfying
\begin{equation*}
\mu_n(A) = \int_{\bxi\in A}\int_{\bx\in[0,1]^n}{\Pb_{n,\bx}^1(\dd^n\bxi)\mu_n(\dd^n\bx)},\quad \text{for all } A\in\mathscr{B}\left([0,1]^n\right);
\end{equation*}
moreover, the finite\hyp{}time distributions, irrespective of the initial distribution at time $k=0$, converge to $\mu_n$ in the total variation norm. One way to obtain results about the stationary distribution $\mu_n$ is thus to fix a convenient initial distribution for $\fitMC_{n,0}$, to determine the finite\hyp{}time distributions, and to analyse their asymptotic behaviour. The most convenient choice of initial distribution turns out to be the uniform distribution $\mathcal{U}_{[0,1]}^{\otimes n}$. We now proceed to give a qualitative description of the finite\hyp{}time distributions.
\begin{proposition}
\label{prop-finitetimedistqual}
For every $n\geq 4$, the finite\hyp{}time distributions of the Markov chain $\fitMC_n$ are absolutely continuous with piecewise polynomial Lebesgue densities. More precisely, for every $k\in\N$ and $\sigma\in\symgr_{n-3}$, there exist finite index sets $I_{n,k,\sigma}\subset\N_0^{n-3}$ and polynomials $q_{n,k,\sigma}\in \Q[x_1,\ldots,x_{n-3}]$ of the form
\begin{equation}
q_{n,k,\sigma}(\bx) = \sum_{\bi\in I_{n,k,\sigma}}\alpha_{n,k,\sigma,\bi}\bx^{\bi},\quad \bx\in[0,1]^{n-3},
\end{equation}
such that
\begin{align}
\label{eq-defgnk}
g_{n,k}(\bx)\dd^n\bx\coloneqq&\Pb\left(\fitMC_{n,k}\in \dd^n\bx\left|\fitMC_{n,0}\sim\mathcal{U}_{[0,1]}^{\otimes n}\right.\right)\notag\\
    =& \left[\sum_{\nu=1}^n\sum_{\sigma\in\symgr_{n-3}}{q_{n,k,\sigma}\left(\bx_{]\nu[}\right)\1_{\{0\leq x_{]\nu[,\sigma(1)}\leq\ldots\leq x_{]\nu[,\sigma(n-3)}\leq 1\}}}\right]\dd^n\bx.
\end{align}
\end{proposition}
\begin{proof}

The finite\hyp{}time distributions are absolutely continuous because the initial distribution at time $k=0$ is, and the transition kernel \ref{eq-transkern} preserves absolute continuity. A reformulation of \cref{eq-recfnkAabstract} shows that the densities $g_{n,k}$ satisfy the recursion
\begin{equation}
\label{eq-recgnk}
g_{n,k+1}(\bx)=\sum_{\nu=1}^n\int_{[0,1]^3} \1_{\left\{\xi_2<\min(\xi_1,\xi_3,\bx_{]\nu[})\right\}}g_{n,k}\left(\bx_{]\nu[_{\bxi}}\right)\dd^3\bxi.
\end{equation}
The claim thus follows by induction on $k$, using the fact that integrals of polynomial expressions over bounded convex polytopes are again polynomial expressions. More precisely, we assume that, for some $k\geq 0$, the function $\bx\mapsto g_{n,k}(\bx)$ is of the form asserted in \cref{eq-defgnk}. It then follows that
\begin{align*}
g_{n,k+1}(\bx) =& \sum_{\nu,\mu,\sigma,\bi}{\alpha_{n,k,\sigma,\bi}}\int_{[0,1]^3}\1_{\left\{\xi_2<\min(\xi_1,\xi_3,\bx_{]\nu[})\right\}}\1_{\left\{\left(\bx_{]\nu[_{\bxi}}\right)_{]\mu[}\sim\sigma\right\}}\left[\left(\bx_{]\nu[_{\bxi}}\right)_{]\mu[}\right]^{\bi}\dd^3\bxi,
\end{align*}
where, for $\by\in[0,1]^{n-3}$ and $\sigma\in\symgr_{n-3}$, the expression $\by\sim\sigma$ is a shorthand for $0\leq y_{\sigma(1)}\leq\ldots\leq y_{\sigma(n-3)}$. Since, for each $\mu=1,\ldots,n$, the integrand is equal to $p\left(\bx_{]\nu[}\right)$ for some polynomial function $p\in\Q[x_1,\ldots,x_{n-3}]$, the claim follows.
\end{proof} 
For the special cases $n=4,5$ we will see later how \cref{eq-recgnk} can be translated into explicit recursions for the coefficients $\alpha_{n,k,\sigma,\bi}$. Before we do that, however, we compute, for general $n$, the density $g_{n,1}$ of the fitness distribution after the first evolution step.
\begin{lemma}
For each $n\geq 4$, the density $g_{n,1}$ is given by
\begin{equation}
g_{n,1}(\bx) = \1_{[0,1]^n}(\bx)\sum_{\nu=1}^n{p(\min\bx_{]\nu[})},\quad p(x) = x\left(\frac13x^2-x+1\right).
\end{equation}
\end{lemma}
\begin{proof}
The claim follows from \cref{eq-recgnk} and the observation that an initial distribution of $\mathcal{U}_{[0,1]}^{\otimes n}$ corresponds to $g_{n,0}(\bx) \equiv 1$. Alternatively, one can use the fact that
\begin{equation*}
\Pb\left(U_2\leq\min(U_1,U_3,y)\right)=\frac{1}{3}\Pb\left(\min(U_1,U_2,U_3)\leq y\right)=\frac{1}{3}\left[1-(1-y)^3\right].\qedhere
\end{equation*}
\end{proof}

\section{Derivation of the asymptotic fitness distribution for four species}
\label{section-n4}
In this section we expand on the results from the last section to compute the stationary distribution in the Bak--Sneppen model with four species, $n=4$. This is much easier than the general situation and we are able to obtain a complete picture of the joint asymptotic behaviour of the four fitness\hyp{}parameters. Partial generalisations to larger populations are discussed in \cref{section-genrec}. For notational convenience we suppress the index $n=4$ in this section. \cref{prop-finitetimedistqual} then implies that there are rational numbers $\alpha_{k,i}$ and positive integers $d_k$ such that the finite\hyp{}time distributions of $\fitMC=\fitMC_4$ are given by
\begin{equation*}
\Pb\left(\fitMC_k\in \dd^4\bx\left|\fitMC_0\sim\mathcal{U}_{[0,1]}^{\otimes 4}\right.\right) = \left[\sum_{\nu=1}^4{q_k(x_\nu)}\right]1_{[0,1]^4}(\bx)\dd^4\bx,\quad q_k(x) = \sum_{i=0}^{d_k}{\alpha_{k,i}x^i},
\end{equation*}
so that $g_k(\bx)=g_{4,k}(\bx)=\sum_i\sum_{\nu=1}^4{\alpha_{k,i}x_\nu^i}$. The Chapman--Kolmogorov\hyp{}type equation \ref{eq-recgnk} implies that, for $k\geq 0$,
\begin{equation*}
\sum_{\nu=1}^4\sum_{i}\alpha_{k+1,i}x_\nu^i = \sum_{\nu=1}^4\sum_i\alpha_{k,i}\int_{[0,1]^3} \1_{\{\xi_2<\min(\xi_1,\xi_3,x_\nu)\}}\left[x_\nu^i+\xi_2^i+\xi_1^i+\xi_3^i\right]\dd^3\bxi.
\end{equation*}
Equating each of the four $\nu$\hyp{}summands separately and evaluating the integrals on the right side, one sees that, for $\nu=1,\ldots,4$,
\begin{align}
\label{eq-comparecoeff4}
&\sum_{i}\alpha_{k+1,i}x_\nu^i \notag\\  
=& \sum_i\alpha_{k,i}\left[x_\nu^i\int_{[0,1]^3} \1_{\{\xi_2<\min(\xi_1,\xi_3,x_\nu)\}}\dd^3\bxi+\int_{[0,1]^3} \1_{\{\xi_2<\min(\xi_1,\xi_3,x_\nu)\}}\xi_2^i\dd^3\bxi+2\int_{[0,1]^3} \1_{\{\xi_2<\min(\xi_1,\xi_3,x_\nu)\}}\xi_1^i\dd^3\bxi\right]\notag\\
  =& \sum_i\alpha_{k,i}x_\nu^{i+1}\left[ \left(\frac13 x_\nu^2-x_\nu+1\right) + \left(\frac{x_\nu^2}{i+3}-\frac{2x_\nu}{i+2}+\frac{1}{i+1}\right) + \frac{2x_\nu}{i+1}\left(\frac{x_\nu}{i+3}-\frac{1}{i+2}\right)\right]+x_\nu(2-x_\nu)\sum_i\frac{\alpha_{k,i}}{i+1}.
\end{align}
Comparing coefficients of equal powers  of $x_\nu$ on both sides of this equation, one arrives at the following recursion for the coefficients $\alpha_{k,i}$.
\begin{proposition}
\label{prop-rec-alpha4}
For each $k\ge 1$, the constant term $\alpha_{k,0}$ vanishes. The remaining coefficients $\alpha_{k,i}$ satisfy the recursion
\begin{equation}
\label{eq-rec-alpha4}
 \alpha_{k+1,i}=\begin{cases} 
2\sum_{j=1}^{3k+1}{\frac{\alpha_{k,j}}{j+1}}, & i=1,\\
\frac{3}{2}\alpha_{k,1}-\sum_{j=1}^{3k+1}{\frac{\alpha_{k,j}}{j+1}}, & i=2,\\
-2\alpha_{k,1}+\frac{4}{3}\alpha_{k,2}, & i=3,\\
\frac{i+1}{i}\alpha_{k,i-1}-\frac{i+1}{i-1}\alpha_{k,i-2}+\frac{1}{3}\frac{i+1}{i-2}\alpha_{k,i-3}, & i\geq 4, 
               \end{cases}\quad k\geq 1.
\end{equation}
In particular, $\alpha_{k,i}$ is equal to zero if $i\geq 3k+1$.
\end{proposition}
\begin{proof}
Only the last statement requires further justification. By the last line of the recursion \ref{eq-rec-alpha4}, the coefficient $\alpha_{k+1,i}$ can be written, for any $\ell=1,\ldots,k+1$, as a linear combination of the coefficients $\alpha_{k+1-\ell,i-m}$, $m=\ell,\ldots,3\ell$. Choosing $\ell=k+1$ and observing that $\alpha_0,i$ is zero for $i\geq 1$, it follows that $\alpha_{k+1.i}$ vanishes for $i\geq 3(k+1)+1$.
\end{proof}

The next step consists in finding as explicit a solution to this recursive system as possible. Fortunately, since we are mainly interested in the asymptotic behaviour of the Bak--Sneppen model, it suffices to concentrate on large values of $k$.
\begin{lemma}
\label{lemma-beta}
For every $i\in\N$, there exists a rational number $\beta_i$ such that $\alpha_{k,i}=\beta_i$ for every $k\geq i+1$. In particular, $\lim_{k\to\infty}\alpha_{k,i} = \beta_i$.
\end{lemma}
\begin{proof}
The claim is true for $i=1$; this follows from the fact that $g_k:\bx\mapsto\sum_\nu q_k(x_\nu)$ is a probability density function and thus integrates to one; by symmetry, $q_k$ then integrates to $1/4$. Observing that $\sum_j{\alpha_{k,j}/(j+1)}$ equals $\int q_k(x)\dd x$, it follows from \cref{eq-rec-alpha4} that, for all $k\geq 1$,
\begin{equation*}
\alpha_{k+1,1} = 2\sum_{j=1}^{3k+1}{\frac{\alpha_{k,j}}{j+1}} = 2\frac14 = \frac12.
\end{equation*}
Similarly, $\alpha_{k,2}=1/2$, $k\geq3$, and $\alpha_{k,3} = -1/3$, $k\geq4$,  and so the claim is true for $i=2,3$ as well. For general $i\geq 4$, the claim follows by induction.
\end{proof}

\begin{proposition}
\label{prop-limitq4}
Let $\beta_i=\alpha_{i+1,i}$ be the numbers from \cref{lemma-beta}. The functions $q_k$ converge uniformly on $[0,1]$, i.\,e.\ in the supremum norm, to $q\coloneqq x\mapsto\sum_i\beta_ix^i$, which is given by
\begin{equation}
q(x) = \frac{3}{2}\frac{x(3-x(3-x))}{(3-x(3-x(3-x)))^2},\quad 0\leq x\leq 1.
\end{equation}
\end{proposition}
\begin{proof}
It follows from \cref{prop-ergodicity} that the functions $q_k$ converge uniformly to some limit, and because $\alpha_{k,i}$ converges, for each $i$, to $\beta_i$ by \cref{lemma-beta}, this limit must be $x\mapsto\sum_i\beta_ix^i$.
Using \cref{prop-rec-alpha4} one sees that the sequence $c_i = (i+1)^{-1}\alpha_{i+1,i}$ satisfies the homogeneous four-term recursion
\begin{equation*}
c_i = c_{i-1} - c_{i-2} + \frac13 c_{i-3},\quad i\geq 4.
\end{equation*}
By the theory of homogeneous linear recurrence equations, $c_i$ is given by
\begin{equation*}
c_i=C_1\lambda_1^i+C_2\lambda_2^i+C_3\lambda_3^i,
\end{equation*}
where
\begin{equation*}
\lambda_1 = \frac{1}{3} \left(1-2^{1/3}+2^{2/3}\right),\quad \lambda_{2,3} = \frac{1}{6} \left(1+2^{1/3}\right) \left(2-2^{1/3}\pm 2^{1/3}\sqrt{3}\ii\right) 
\end{equation*}
are the roots of the characteristic polynomial $z^3-z^2+z-1/3$, and $C_j$ are constants determined by the initial conditions. In the present case, as the proof of \cref{lemma-beta} shows, these initial values are given by
\begin{equation*}
c_1 = \frac{1}{2}\alpha_{2,1} = \frac{1}{4},\quad c_2 = \frac{1}{3}\alpha_{3,2} = \frac{1}{6},\quad c_3 = \frac{1}{4}\alpha_{4,3} = -\frac{1}{12}.
\end{equation*}
They can also easily be read off \cref{table-alpha4}, where we have recorded the functions $q_k$ for $k=0,\ldots,4$. The constants $C_i$ are thus the solution of the Vandermonde system
\begin{align*}
c_i =& C_1\lambda_1^i+C_2\lambda_2^i + C_3\lambda_3^i,\quad i=1,2,3,
\end{align*}
that is
\begin{align*}
C_1 =& \frac{c_3+c_1 \lambda_2 \lambda_3-c_2 \left(\lambda_2+\lambda_3\right)}{\lambda_1 \left(\lambda_1-\lambda_2\right) \left(\lambda_1-\lambda_3\right)},\quad C_2 = -\frac{c_3+c_1 \lambda_1 \lambda_3-c_2 \left(\lambda_1+\lambda_3\right)}{\left(\lambda_1-\lambda_2\right) \lambda_2 \left(\lambda_2-\lambda_3\right)},\quad C_3 = \frac{c_3+c_1 \lambda_1 \lambda_2-c_2 \left(\lambda_1+\lambda_2\right)}{\left(\lambda_2-\lambda_3\right) \lambda_3 \left(\lambda_1 - \lambda_3\right)}.
\end{align*}
It is then easy to compute $q$ as
\begin{align*}
q(x) =& \sum_{i=1}^\infty{\beta_ix^i}  = \sum_{j=1}^3 C_j\sum_{i=1}^\infty(i+1)\left(x\lambda_j\right)^i = \sum_{j=1}^3C_j\frac{x\lambda_j(2-x\lambda_j)}{(1-x\lambda_j)^2}=\frac{3}{2}\frac{x(3-x(3-x))}{(3-x(3-x(3-x)))^2}.
\end{align*}
\end{proof}
\begin{remark}
The form of the function $q$ with all Horner coefficients being equal to $3$ is very interesting. We do not have an intuitive explanation for this peculiarity.
\end{remark}

\begin{table}
\begin{tabular}{c|c|cccccccccccccc}
			& 	& \multicolumn{14}{c}{$i$}																				\\ \hline
			& 	& $0$	& $1$		& $2$		& $3$		& $4$	& $5$	& $6$		& $7$	& $8$		& $9$		& $10$		& $11$		& $12$ 	 	\\ \hline
\multirow{5}{*}{$k$} 	& $0$ 	& $\frac14$	& $0$																									\\
			& $1$ 	& $0$	& $1$		& $-1$ 		& $\frac{1}{3}$		& $0$																\\
			& $2$ 	& 	& $\boxed{\mathbf{\frac{1}{2}}}$& $\frac{5}{4}$		& $-\frac{10}{3}$ 	&$\frac{35}{12}$& $-\frac{7}{6}$&$\frac{7}{36}$		& $0$												\\
			& $3$ 	& 	& $\frac{1}{2} $	& $\boxed{\mathbf{\frac{1}{2}}}$& $\frac{2}{3}$		&$-\frac{35}{6}$& $\frac{28}{3}$&$-\frac{133}{18}$	&$\frac{10}{3}$	&$-\frac{5}{6}$		& $\frac{5}{24}$	& $0$							\\
			& $4$ 	& 	& $\frac{1}{2}$		& $\frac{1}{2}$		&$\boxed{\mathbf{-\frac{1}{3}}}$& $\frac{5}{12}$&$-\frac{23}{3}$&$\frac{175}{9}$	&$-24$	&$\frac{215}{12}$	&$-\frac{188}{18}$	&$\frac{143}{54}$	&$-\frac{13}{27}$	& $\frac{13}{324}$& $0$	
\end{tabular}
\vspace{20pt}
\caption{Values of the coefficients $\alpha_{k,i}$ determining the finite\hyp{}time distributions of the four\hyp{}species Bak--Sneppen model. Values of $\beta_i$, $i=1,2,3$, emphasised. Empty cells indicate zeros.}\label{table-alpha4}
\end{table}

We have thus proved the following result.
\begin{theorem}
\label{theorem-asympdist4}
Let $q$ be as in \cref{prop-limitq4}. The stationary distribution $\mu=\mu_4$ of the Bak--Sneppen model with four species is given by
\begin{equation}
\mu\left(\dd^4\bx\right) = \left[\sum_{\nu=1}^4 q(x_\nu)\right]\1_{[0,1]^4}(\bx)\dd^4\bx.
\end{equation}
In particular, the one\hyp{}dimensional marginal of $\mu$, which is the asymptotic fitness distribution at a fixed site, is absolutely continuous with density
\begin{equation}
\label{eq-asympmargdens}
\frac{\dd}{\dd x}\Pb(f_{1,\infty}\leq x) = \left(\frac{3}{4}+q(x)\right)\1_{[0,1]}(x).
\end{equation}
\end{theorem}
In \cref{fig-margdens} we have plotted the densities of the one\hyp{}dimensional marginal distributions of $\fitMC_k$ for $k=0,\ldots,6$, together with their limit as given by \cref{eq-asympmargdens}. The convergence asserted by \cref{prop-limitq4} is clearly visible.
\begin{figure}
\center
\includegraphics[width=.5\textwidth]{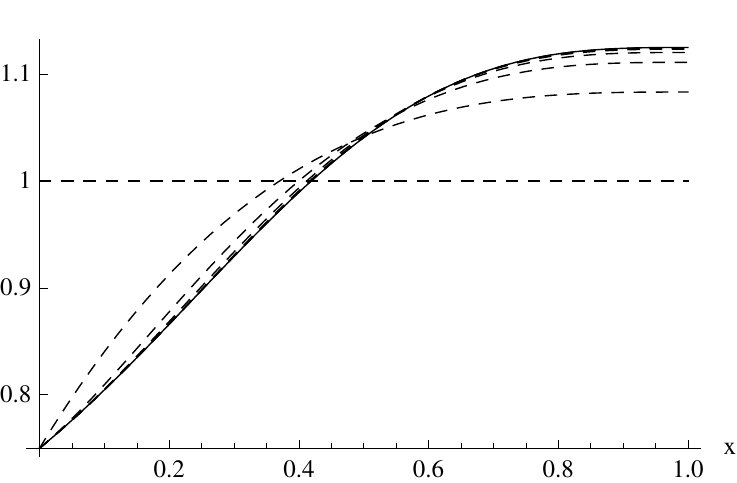}
\caption{Plot of the densities of the one\hyp{}dimensional marginal distributions of $\fitMC_{4,k}$ for $k=0,\ldots,6$ (dashed), together with their limit (solid) as given by \cref{eq-asympmargdens}}
\label{fig-margdens}
\end{figure}

\section{Further steps towards a solution for general population sizes}
\label{section-genrec}
For general values of $n$, a similarly complete analysis of the asymptotic fitness distribution in the BS model as we presented for $n=4$ in the previous section seems difficult to obtain. Here we collect some formulae that serve as the starting point if one wants to carry through the same approach as in \cref{section-n4}. Similarly to \cref{eq-comparecoeff4}, it holds for every $\nu=1,\ldots,n$ that
\begin{align}
\label{eq-comparecoeffn}
  & \sum_{\sigma,\bi}\alpha_{n,k+1,\sigma,\bi}\bx_{]\nu[}^{\bi}\1_{\{0\leq \bx_{]\nu[,\sigma(1)}\leq\ldots\leq \bx_{]\nu[,\sigma(n-3)}\}}\notag\\
= & \sum_{\mu,\bi,\sigma}\alpha_{n,k,\sigma,\bi}\int_{[0,1]^3} \1_{\left\{\xi_2<\min(\xi_1,\xi_3,\bx_{]\nu[})\right\}}\left[\left(\bx_{]\nu[_{\bxi}}\right)_{]\mu[}\right]^{\bi}\1_{\left\{0\leq\left(\bx_{]\nu[_{\bxi}}\right)_{]\mu[,\sigma(1)}\leq\ldots\leq\left(\bx_{]\nu[_{\bxi}}\right)_{]\mu[,\sigma(n-3)}\leq 1\right\}}\dd^3\bxi.
\end{align}
The integrals that appear on the right hand side of this equation fall in four categories, depending on how many of the three $\xi$'s are hit by the outer projection $(\cdot)_{]\mu[}$.
\paragraph{Case 1: $\mu=2$.}
In this case, all three $\xi$'s are hit by the outer projection and the integral on the right side of \cref{eq-comparecoeffn} simplifies to
\begin{align*}
&\int_{[0,1]^3} \1_{\left\{\xi_2<\min(\xi_1,\xi_3,\bx_{]\nu[})\right\}}\left(\bx_{]\nu[}\right)^{\bi}\1_{\left\{0\leq\bx_{]\nu[,\sigma(1)}\leq\ldots\leq\bx_{]\nu[,\sigma(n-3)}\leq 1\right\}}\dd^3\bxi\\
=&\left(\bx_{]\nu[}\right)^{\bi}\bx_{]\nu[,\sigma(1)}\left[\frac13\bx_{]\nu[,\sigma(1)}^2 - \bx_{]\nu[,\sigma(1)}^2 + 1\right]\1_{\left\{0\leq\bx_{]\nu[,\sigma(1)}\leq\ldots\leq\bx_{]\nu[,\sigma(n-3)}\leq 1\right\}}.  
\end{align*}
\paragraph{Case 2: $\mu\in\{1,3\}$.}
In this case, either $\xi_1$ (if $\mu=3$) or $\xi_3$ (if $\mu=1$) is not hit by the outer $\mu$\hyp{}projection. For definiteness, we assume that $\mu=3$, the other sub\hyp{}case being analogous. The sequential application of the outer projections $(\cdot)_{]\nu[_{\bxi}}$ and $(\cdot)_{]3[}$ to $\bx$ results in the vector $\by=(x_{\nu+3},\ldots,x_n,x_1,\ldots,x_{\nu-2},\xi_1)$, contributing a factor $\xi_1^{i_{n-3}}$ to the integral. With the definitions $\by_{\sigma(0)}\coloneqq 0$, $\by_{\sigma(n-2)}\coloneqq 1$ the integral in \cref{eq-comparecoeffn} thus becomes
\begin{align}
\label{eq-case2}
&\by_{1:n-4}^{\bi_{1:n-4}}\left[\int_{[0,1]^3} \1_{\left\{\xi_2<\min(\xi_1,\xi_3,\bx_{]\nu[})\right\}}\1_{\left\{\by_{\sigma\left(\sigma^{-1}(n-3)-1\right)}\leq\xi_1\leq\by_{\sigma\left(\sigma^{-1}(n-3)+1\right)}\right\}}\xi_1^{i_{n-3}}\dd^3\bxi\right]\1^{(*)}_{\left\{0\leq\by_{\sigma(1)}\leq\ldots\leq\by_{\sigma(n-3)}\leq 1\right\}},
\end{align}
where $\by_{1:n-4}=(x_{\nu+3},\ldots,x_n,x_1,\ldots,x_{\nu-2})$, $\bi_{1:n-4}=(i_1,\ldots,i_{n-4})$, and the superscript $(*)$ indicates that only the ordering of the components of $\by$ that pertain to $\bx$, in contrast to $\bxi$, is restricted. Integration of this expression is standard and left to the reader.
\paragraph{Case 3: $\mu\in\{n,4\}$.}
In this case, both $\xi_2$ and either $\xi_3$ (if $\mu=4$) or $\xi_1$ (if $\mu=n$) are hit by the outer $\mu$\hyp{}projection. Again, for definiteness, we assume that $\mu=4$, and observe that the other sub\hyp{}case can be treated in an analogous way. In this case
\begin{equation*}
\left(\bx_{]\nu[_{\bxi}}\right)_{]4[} = \by\coloneqq\left(x_{\nu+4},\ldots,x_n,x_1,\ldots,x_{\nu-2},\xi_1,\xi_2\right),
\end{equation*}
resulting in a factor $\xi_1^{i_{n-4}}\xi_2^{i_{n-3}}$ in the integrand; the integral in \cref{eq-comparecoeffn} is clearly zero unless $\sigma(1)=n-3$, so as not contradict the first indicator function that specifies $\xi_2$ to be less than $\xi_1$ and $\bx_{]\nu[}$. With the obvious definitions of $\by_{1;n-5}$ and $\bi_{1:n-5}$, the integral then becomes
\begin{align}
\label{eq-case3}
&\by_{1:n-5}^{\bi_{1:n-5}}\left[\int_{[0,1]^3} \1_{\left\{\xi_2<\min(\xi_1,\xi_3,\bx_{]\nu[})\right\}}\1_{\left\{\by_{\sigma\left(\sigma^{-1}(n-4)-1\right)}\leq\xi_1\leq\by_{\sigma\left(\sigma^{-1}(n-4)+1\right)}\right\}}\xi_1^{i_{n-4}}\xi_2^{i_{n-3}}\dd^3\bxi\right]\1^{(*)}_{\left\{0\leq\by_{\sigma(2)}\leq\ldots\leq\by_{\sigma(n-3)}\leq 1\right\}},
\end{align}
the evaluation of which we leave to the reader.
\paragraph{Case 4: all other $\mu$.}
In this last case,  all three $\xi$'s remain unaffected by the $\mu$\hyp{}projection and contribute to making the integral a little more complicated than before. Similar to the previous cases one obtains
\begin{equation*}
\by\coloneqq\left(\bx_{]\nu[_{\bxi}}\right)_{]\mu[} = \left(x_{\mu+2},\ldots,x_n,x_1,\ldots,x_{\nu-2},\xi_1,\xi_2,\xi_3,x_{\nu+2},\ldots,x_{\mu-2}\right),
\end{equation*}
and, again, the integral vanishes if $\by_{\sigma(1)}\neq\xi_2$. Consequently, the integral in \cref{eq-comparecoeffn} becomes
\begin{align}
\label{eq-case4}
&\by_*^{\bi_*}\left[\int_{[0,1]^3} \1_{\left\{\xi_2<\min(\xi_1,\xi_3,\bx_{]\nu[})\right\}}\1_{\left\{\by_{\sigma\left(\sigma^{-1}(j_1)-1\right)}\leq\xi_1\leq\by_{\sigma\left(\sigma^{-1}(j_1)+1\right)}\right\}}\1_{\left\{\by_{\sigma\left(\sigma^{-1}(\nu+1)-1\right)}\leq\xi_3\leq\by_{\sigma\left(\sigma^{-1}(\nu+1)+1\right)}\right\}}\right.\notag\\
&\qquad\qquad\times\left.\xi_1^{i_{j_1}}\xi_2^{i_{j_2}}\xi_3^{i_{j_3}}\dd^3\bxi\right]\1^{(*)}_{\left\{0\leq\by_{\sigma(2)}\leq\ldots\leq\by_{\sigma(n-3)}\leq 1\right\}},
\end{align}
where $\by_*=\left(x_{\mu+2},\ldots,x_n,x_1,\ldots,x_{\nu-2},x_{\nu+2},\ldots,x_{\mu-2}\right)$, $\by_{j_k}=\xi_k$, $k=1,2,3$, and $\bi_*=(i_j)_{j\neq j_k}$. According to whether or not $\xi_1$ and $\xi_3$ are separated by some $x_j$ in $(\by_{\sigma(j)})_j$, there are a number of sub\hyp{}cases to consider, which, however, do not present any difficulty. Closed\hyp{}form expressions of the integrals \ref{eq-case2}, \ref{eq-case3}, and \ref{eq-case4} have been obtained, but are omitted for lack of space and because they do not add much insight at the current stage.

By equating equal powers of $\bx$, \cref{eq-comparecoeffn}, can thus, in principle, be used to derive an explicit recursion for the coefficients $\alpha_{n,k,\sigma,\bi}$. Here, we only illustrate this potential by giving the recursion for $n=5$, leaving the problem of working out the combinatorics of the general case for future research. An easy symmetry argument shows that $\alpha_{5,k,(1,2),(i_1,i_2)}$ is equal to $\alpha_{5,k,(2,1),(i_2,i_1)}$ so that it is enough to consider the identity permutation $\mathrm{id}\in\symgr_2$.
\begin{proposition}
\label{prop-rec-alpha5}
For $n=5$, the coefficients $\alpha_{k,i,j}\coloneqq\alpha_{5,k,\mathrm{id},\bi}$ have the following properties:
\begin{enumerate}[i)]
 \item $\alpha_{k,0,j}=0$ for all $k\geq 1$ and $j\geq 0$.
 \item For $k\geq 1$ and $j\geq 1$, they satisfy the recursion
\begin{equation}
\label{eq-recn5jgeq1}
 \alpha_{k+1,i,j} = \begin{cases}
\sum_{p=0}^{3k+1}\frac{\alpha_{k,j,p}}{p+1}+\sum_{p=0}^{j-1}\frac{2p-j+1}{(p+1)(j-p)}\alpha_{k,j-1-p,p},		      & i=1,\\
\alpha_{k,1,j}-\frac{1}{2}\alpha_{k+1,i,j},		      & i=2,\\
\alpha_{k,i-1,j}-\left[1+\frac{1}{i(i-1)}\right]\alpha_{k,i-2,j}+\left[\frac{1}{3}+\frac{1}{i(i-2)}\right]\alpha_{k,i-3,j},		      & i\geq3.
                     \end{cases}
\end{equation}
\item For $k\geq 1$ and $j=0$, they satisfy $\alpha_{k,1,0}=0$ as well as the recursion
\begin{equation}
\label{eq-recn5j0}
\alpha_{k+1,i,0} = \begin{cases}
2\sum_{p=0}^{3k+1}\frac{\alpha_{k,1,p}}{p+1},  & i=2,\\
\alpha_{k,i-1,0}-\left[1+\frac{1}{(i-1)i}\right]\alpha_{k,i-2,0}+\left[\frac{1}{3}+\frac{1}{(i-2)i}\right]\alpha_{k,i-3,0}\\
\quad+\left(1+\frac{2}{i}\right)\sum_{p=0}^{3k+1}\frac{\alpha_{k,i-1,p}}{p+1}-\left(\frac{1}{2}+\frac{2}{i}\right)\sum_{p=0}^{3k+1}\frac{\alpha_{k,i-2,p}}{p+1}\\
\quad-\left(1+\frac{2}{i}\right)\sum_{p=0}^{i-2}\frac{\alpha_{k,i-2-p,p}}{p+1}+\left(\frac{1}{2}+\frac{2}{i}\right)\sum_{p=0}^{i-3}\frac{\alpha_{k,i-3-p,p}}{p+1}\\
\quad+\sum_{p=0}^{i-2}\frac{\alpha_{k,i-2-p,p}}{i-p} -\frac{1}{2} \sum_{p=0}^{i-3}\frac{\alpha_{k,i-3-p,p}}{i-p},		      & i\geq3.
                     \end{cases}
\end{equation}
\end{enumerate}
\end{proposition}
Unfortunately, it does not seem possible to find an explicit solution to this recursion. In fact, the algorithm HYPER for linear recurrence relations with polynomial coefficients \citep{Petkovsek1992hypergeometric} indicates that the last line of \cref{eq-recn5jgeq1} does not have a hypergeometric solution. Our results can, however, be used to obtain exact expression for the finite\hyp{}time fitness distributions in the five\hyp{}species BS model without having to evaluate any integrals. Calculations based on \cref{eq-recn5jgeq1} and \cref{eq-recn5j0} give rise to the following conjecture.
\begin{conjecture}
The coefficients $\alpha_{k,i,j} = \alpha_{5,k,\mathrm{id},\bi}$ have the following properties:
\begin{enumerate}[i)]
 \item They form a stair case pattern in the sense that, for $j\geq 0$,
\begin{align*}
 \alpha_{k,i,0} = 0,\quad & i\geq 3k,\\
 \alpha_{k,i,1} = 0,\quad & i\geq 3(k-1),\\
  \alpha_{k,i,j} =0,\quad & i\geq 3(k-\lambda),\quad \lambda=\left\lfloor\frac{j+1}{3}\right\rfloor;
\end{align*}
in particular, $\alpha_{k,1,j}=0$ for $j\geq 3k-1$.
\item They become constant at time $i+j+1$, that is there exist numbers $\beta_{i,j}$ such that $\alpha_{k,i,j}=\beta_{i,j}$ for all $k\geq i+j+1$.
\end{enumerate}
\end{conjecture}
In fact, we conjecture that, for each $n$, there exist $k_{n,\sigma,\bi}\in\N$ and $\beta_{n,\sigma,\bi}\in\Q$ such that $\alpha_{n,k,\sigma,\bi}=\beta_{n,\sigma,\bi}$ for all $k\geq k_{n,\sigma,\bi}$.
\section*{Acknowledgements}
I thank two anonymous referees for carefully reading the manuscript and making helpful suggestions to improve its presentation. I also thank them for pointing out the applicability of our method to the anisotropic Bak--Sneppen model.

\end{document}